\documentclass[11pt]{article}
\usepackage{graphicx, subfigure}
\usepackage{epsf}
\usepackage{epsfig}
\usepackage{latexsym}
\usepackage{times}
\usepackage{mathptm}
\usepackage{defs}

\newtheorem{lemma}{Lemma}
\newtheorem{theorem}{Theorem}
\newenvironment{proof}{\textit{Proof:}~}{\hfill$\Box$\par\vskip1em}

\newcommand{\Omit}[1]{}
\newcommand{\CONF}{{\cal C}}
\begin{document}

\title{Asynchronous mobile robot gathering from symmetric configurations without global multiplicity detection}

\author{Sayaka Kamei\footnote{
  This work is supported in part by a Grant-in-Aid for Young
  Scientists ((B)22700074) of JSPS.}\\
   Dept. of Information Engineering, Hiroshima University, Japan \\
 s-kamei@se.hiroshima-u.ac.jp \\
   \and
   Anissa Lamani \\
   Universit\'{e} de picardie Jules Vernes, Amiens, France \\
  anissa.lamani@u-picardie.fr\\
  \and
   Fukuhito Ooshita \\
  Graduate School of Information Science and Technology, Osaka University, Suita, 565-0871, Japan \\
  f-oosita@ist.osaka-u.ac.jp\\
  \and
   S\'{e}bastien Tixeuil \\
  LIP6 UMR 7606, Universit\'{e} Pierre et Marie Curie, France \\
  Sebastien.Tixeuil@lip6.fr
   }

\date{}
\maketitle

\begin{abstract}
We consider a set of $k$ autonomous robots that are endowed with visibility sensors (but that are otherwise unable to communicate) and motion actuators. Those robots must collaborate to reach a single vertex that is unknown beforehand, and to remain there hereafter. Previous works on gathering in ring-shaped networks suggest that there exists a tradeoff between the size of the set of potential initial configurations, and the power of the sensing capabilities of the robots (\emph{i.e.} the larger the initial configuration set, the most powerful the sensor needs to be).
We prove that there is no such trade off. We propose a gathering protocol for an odd number of robots in a ring-shaped network that allows symmetric but not periodic configurations as initial configurations, yet uses only local weak multiplicity detection. Robots are assumed to be anonymous and oblivious, and the execution model is the non-atomic CORDA model with asynchronous fair scheduling. Our protocol allows the largest set of initial configurations (with respect to impossibility results) yet uses the weakest multiplicity detector to date. The time complexity of our protocol is $O(n^2)$, where $n$ denotes the size of the ring. Compared to previous work that also uses local weak multiplicity detection, we do \emph{not} have the constraint that $k<n/2$ (here, we simply have $2<k<n-3$).
\end{abstract}

%
\centerline{{\bf Keywords}: Gathering, Discrete Universe, Local Weak Multiplicity Detection, Asynchrony, Robots. }

\section{Introduction}

We consider autonomous robots that are endowed with visibility sensors (but that are otherwise unable to communicate) and motion actuators. Those robots must collaborate to solve a collective task, namely \emph{gathering}, despite being limited with respect to input from the environment, asymmetry, memory, etc. The area where robots have to gather is modeled as a graph and the gathering task requires every robot to reach a single vertex that is unknown beforehand, and to remain there hereafter.

Robots operate in \emph{cycles} that comprise \emph{look}, \emph{compute}, and \emph{move} phases. The look phase consists in taking a snapshot of the other robots positions using its visibility sensors. In the compute phase a robot computes a target destination among its neighbors, based on the previous observation. The move phase simply consists in moving toward the computed destination using motion actuators. We consider an asynchronous computing model, i.e., there may be a finite but unbounded time between any two phases of a robot's cycle. Asynchrony makes the problem hard since a robot can decide to move according to an old snapshot of the system and different robots may be in different phases of their cycles at the same time.

Moreover, the robots that we consider here have weak capacities: they are \emph{anonymous} (they execute the same protocol and have no mean to distinguish themselves from the others), \emph{oblivious} (they have no memory that is persistent between two cycles), and have no compass whatsoever (they are unable to agree on a common direction or orientation in the ring). 

\subsection{Related Work}\label{sec:RW}

While the vast majority of literature on coordinated distributed robots considers that those robots are evolving in a \emph{continuous} two-dimensional Euclidean space and use visual sensors with perfect accuracy that permit to locate other robots with infinite precision, a recent trend was to shift from the classical continuous model to the \emph{discrete} model. In the discrete model, space is partitioned into a \emph{finite} number of locations. This setting is conveniently represented by a graph, where nodes represent locations that can be sensed, and where edges represent the possibility for a robot to move from one location to the other. Thus, the discrete model restricts both sensing and actuating capabilities of every robot. For each location, a robot is able to sense if the location is empty or if robots are positioned on it (instead of sensing the exact position of a robot). Also, a robot is not able to move from a position to another unless there is explicit indication to do so (\emph{i.e.}, the two locations are connected by an edge in the representing graph). The discrete model permits to simplify many robot protocols by reasoning on finite structures (\emph{i.e.}, graphs) rather than on infinite ones. However, as noted in most related papers~\cite{Marco06,Kowalski04,davi07,Devismes09,Lamani10,Blin10,Izumi10,Flocchin04,Klasing08,Klasing06}, this simplicity comes with the cost of extra symmetry possibilities, especially when the authorized paths are also symmetric.

In this paper, we focus on the discrete universe where two main problems have been investigated under these weak assumptions. The \textit{exploration problem} consists in exploring a given graph using a minimal number of robots. Explorations come in two flavours: \emph{with stop} (at the end of the exploration all robots must remain idle) \cite{davi07,Devismes09,Lamani10} and \emph{perpetual} (every node is visited infinitely often by every robot) \cite{Blin10}. The second studied problem is the \emph{gathering} problem where a set of robots has to gather in one single location, not defined in advance, and remain on this location \cite{Izumi10,Flocchin04,Klasing08,Klasing06}. 

The gathering problem was well studied in the continuous model with various assumptions \cite{Ciel03,Ciel04,Flocchin05,Prencipe05}. In the discrete model, deterministic algorithms have been proposed to solve the gathering problem in a ring-shaped network, which enables many problems to appear due to the high number of symmetric configurations. In \cite{Marco06,Kowalski04,Dessmark06}, symmetry was broken by enabling robots to distinguish themselves using labels, in \cite{Flocchin04}, symmetry was broken using tokens. The case of anonymous, asynchronous and oblivious robots was investigated only recently in this context. It should be noted that if the configuration is periodic and edge symmetric, no deterministic solution can exist~\cite{Klasing06}. The first two solutions \cite{Klasing06,Klasing08} are complementary: \cite{Klasing06} is based on breaking the symmetry whereas \cite{Klasing08} takes advantage of symmetries. However, both \cite{Klasing06} and \cite{Klasing08} make the assumption that robots are endowed with the ability to distinguish nodes that host one robot from nodes that host two robots or more in the entire network (this property is referred to in the literature as \emph{global} weak multiplicity detection). This ability weakens the gathering problem because it is sufficient for a protocol to ensure that a single multiplicity point exists to have all robots gather in this point, so it reduces the gathering problem to the creation of a single multiplicity point.

Investigating the feasibility of gathering with weaker multiplicity detectors was recently addressed in \cite{Izumi10}. In this paper, robots are only able to test that their current hosting node is a multiplicity node (\emph{i.e.} hosts at least two robots). This assumption (referred to in the literature as \emph{local} weak multiplicity detection) is obviously weaker than the global weak multiplicity detection, but is also more realistic as far as sensing devices are concerned. The downside of \cite{Izumi10} compared to \cite{Klasing08} is that only rigid configurations (\emph{i.e.} non symmetric configuration) are allowed as initial configurations (as in  \cite{Klasing06}), while \cite{Klasing08} allowed symmetric but not periodic configurations to be used as initial ones. Also, \cite{Izumi10} requires that $k < n/2$ even in the case of non-symmetric configurations.

\subsection{Our Contribution}

We propose a gathering protocol for an odd number of robots in a ring-shaped network that allows symmetric but not periodic configurations as initial configurations, yet uses only local weak multiplicity detection. Robots are assumed to be anonymous and oblivious, and the execution model is the non-atomic CORDA model with asynchronous fair scheduling. Our protocol allows the largest set of initial configurations (with respect to impossibility results) yet uses the weakest multiplicity detector to date. The time complexity of our protocol is $O(n^2)$, where $n$ denotes the size of the ring. By contrast to \cite{Izumi10}, $k$ may be greater than $n/2$, as our constraint is simply that $2<k<n-3$ and $k$ is odd.

\section{Preliminaries}

\subsection{System Model}\label{sec:Model}
We consider here the case of an anonymous, unoriented and undirected ring of $n$ nodes $u_{0}$,$u_{1}$,..., $u_{(n-1)}$ such as $u_{i}$ is connected to both $u_{(i-1)}$ and $u_{(i+1)}$. 
Note that since no labeling is enabled (anonymous), there is no way to distinguish between nodes, or between edges. 

On this ring, $k$ robots operate in distributed way in order to accomplish a common task that is to gather in one location not known in advance. 
We assume that $k$ is odd.
The set of robots considered here are \textit{identical}; they execute the same program using no local parameters and one cannot distinguish them using their appearance, and are \textit{oblivious}, which means that they have no memory of past events, they can't remember the last observations or the last steps taken before. 
In addition, they are unable to communicate directly, however, they have the ability to sense the environment including the position of the other robots. 
Based on the configuration resulting of the sensing, they decide whether to move or to stay idle. 
Each robot executes cycles infinitely many times, 
~(1)~first, it catches a sight of the environment to see the position of the other robots (look phase), 
~(2)~according to the observation, it decides to move or not (compute phase), 
~(3)~if it decides to move, it moves to its neighbor node towards a target destination (move phase). 

At instant $t$, a subset of robots are activated by an entity known as \textit{the scheduler}.
The scheduler can be seen as an external entity that selects some robots for execution, this scheduler is considered to be fair, which means that, all robots must be activated infinitely many times. 
The \textit{CORDA model} \cite{Pre01} enables the interleaving of phases by the scheduler 
(For instance, one robot can perform a look operation while another is moving).
The model considered in our case is the CORDA model with the following constraint: the Move operation is instantaneous \textit{i.e.} when a robot takes a snapshot of its environment, it sees the other robots on nodes and not on edges. However, since the scheduler is allowed to interleave the different operations, robots can move according to an outdated view since during the Compute phase, some robots may have moved.

During the process, some robots move, and at any time occupy nodes of the ring, their positions form a configuration of the system at that time. 
We assume that, at instant $t=0$ (\textit{i.e.}, at the initial configuration), some of the nodes on the ring are occupied by robots, such as, each node contains at most one robot. 
If there is no robot on a node, we call the node \textit{empty node}.
The segment $[u_p,u_q]$ is defined by the sequence $(u_p,u_{p+1},\cdots,u_{q-1},u_q)$ of consecutive nodes in the ring, 
such as all the nodes of the sequence are empty except $u_p$ and $u_q$ that contain at least one robot.
The distance $D_p^t$ of segment $[u_p,u_q]$ in the configuration of time $t$ is equal to the number of nodes in $[u_p,u_q]$ minus $1$.
We define a \textit{hole} as the maximal set of consecutive empty nodes. 
That is, in the segment $[u_p,u_q]$, $(u_{p+1},\cdots,u_{q-1})$ is a hole. 
The size of a hole is the number of free nodes that compose it, the border of the hole are the two empty nodes who are part of this hole, having one robot as a neighbor.

We say that there is a \textit{tower} at some node $u_{i}$, if at this node there is more than one robot (Recall that this tower is distinguishable only locally). 

When a robot takes a snapshot of the current configuration on node $u_{i}$ at time $t$, it has a \textit{view} of the system at this node.
In the configuration $C(t)$, we assume $[u_1,u_2],[u_2,u_3],\cdots,[u_w,u_1]$ are consecutive segments in a given direction of the ring.
Then, the view of a robot on node $u_1$ at $C(t)$ is represented by \\$(\max\{(D_1^t,D_2^t,\cdots,D_w^t),(D_w^t,D_{w-1}^t,\cdots,D_1^t)\}, m_1^t)$, where
$m_1^t$ is true if there is a tower at this node, and sequence $(a_i,a_{i+1},\cdots, a_j)$ is larger than $(b_i,b_{i+1},\cdots,b_j)$ if there is $h (i\leq h\leq j)$ such that $a_l=b_l$ for $i\leq l \leq h-1$ and $a_h>b_h$.
It is stressed from the definition that robots don't make difference between a node containing one robot and those containing more. 
However, they can detect $m^t$ of the current node, \textit{i.e.} whether they are alone on the node or not (they have a local weak multiplicity detection). 

When $(D_1^t,D_2^t,\cdots,D_w^t)=(D_w^t,D_{w-1}^t,\cdots,D_1^t)$, we say that the view on $u_i$ is \textit{symmetric}, otherwise we say that the view on $u_{i}$ is \textit{asymmetric}. 
Note that when the view is symmetric, both edges incident to $u_{i}$ look identical to the robot located at that node. 
In the case the robot on this node is activated we assume the worst scenario allowing the scheduler to take the decision on the direction to be taken.  

Configurations that have no tower are classified into three classes in \cite{Klasing08-j}.
Configuration is called \textit{periodic} if it is represented by a configuration of at least two copies of a sub-sequence.
Configuration is called \textit{symmetric} if the ring contains a single axis of symmetry.
 Otherwise, the configuration is called \textit{rigid}. 
For these configurations, the following lemma is proved in \cite{Klasing06}.
\begin{lemma}
If a configuration is rigid, all robots have distinct views. If a configuration is symmetric and non-periodic, there exists exactly one axis of symmetry.
\end{lemma}
This lemma implies that, if a configuration is symmetric and non-periodic, at most two robots have the same view.

We now define some useful terms that will be used to describe our algorithm. 
We denote by the \textit{inter-distance} $d$ the minimum distance taken among distances between each pair of distinct robots (in term of the number of edges). 
Given a configuration of inter-distance $d$, a \textit{$d$.block} is any maximal elementary path where there is a robot every $d$ edges. 
The border of a $d$.block are the two external robots of the $d$.block. 
The size of a $d$.block is the number of robots that it contains. 
We call the $d$.block whose size is biggest the \textit{biggest $d$.block}.
A robot that is not in any $d$.block is said to be an \textit{isolated robot}. 

We evaluate the time complexity of algorithms with asynchronous rounds. An asynchronous round is defined as the shortest fragment of an execution where each robot performs a move phase at least once.

\subsection{Problem to be solved}\label{sec:Prb}
The problem considered in our work is the gathering problem, where $k$ robots have to agree on one location (one node of the ring) not known in advance in order to gather on it, and that before stopping there forever.

\section{Algorithm}\label{sec:Algo}


To achieve the gathering, we propose the algorithm composed of two phases. The first phase is to build a configuration that contains a single $1$.block and no isolated robots without creating any tower regardless of their positions in the initial configuration (provided that there is no tower and the configuration is aperiodic.)
The second phase is to achieve the gathering from any configuration that contains a single $1$.block and no isolated robots. 
Note that, since each robot is oblivious, it has to decide the current phase by observing the current configuration. 
To realize it, we define a special configuration set $\CONF_{sp}$ which includes any configuration that contains a single $1$.block and no isolated robots. 
We give the behavior of robots for each configuration in $\CONF_{sp}$, and guarantee that the gathering is eventually achieved from any configuration in $\CONF_{sp}$ without moving out of $\CONF_{sp}$. We combine the algorithms for the first phase and the second phase in the following way: Each robot executes the algorithm for the second phase if the current configuration is in $\CONF_{sp}$, and executes one for the first phase otherwise. By this way, as soon as the system becomes a configuration in $\CONF_{sp}$ during the first phase, the system moves to the second phase and the gathering is eventually achieved.

\subsection{First phase: An algorithm to construct a single $1$.block}

In this section, we provide the algorithm for the first phase, that is, the algorithm to construct a configuration with a single $1$.block. The strategy is as follows; 
In the initial configuration, robots search the biggest $d$.block $B_1$, and then robots that are not on $B_1$ move to join $B_1$.
Then, we can get a single $d$.block.
In the single $d$.block, there is a robot on the axis of symmetry because the number of robots is odd.
When the nearest robots from the robot on the axis of symmetry move to the axis of symmetry, then we can get a $d-1$.block $B_2$,
and robots that are not on $B_2$ move toward $B_2$ and join $B_2$.
By repeating this way, we can get a single $1$.block. 

We will distinguish three types of configurations as follows:
\begin{itemize*}
\item \textbf{Configuration of type $1$}.
In this configuration, there is only a single $d$.block such as $d > 1$, that is, all the robots are part of the $d$.block. 
Note that the configuration is in this case symmetric, and since there is an odd number of robots, we are sure that there is one robot on the axis of symmetry.

If the configuration is this type, the robots that are allowed to move are the two symmetric robots that are the closest to the robot on the axis. 
Their destination is their adjacent empty node towards the robot on the axis of symmetry.
(Note that the inter-distance has decreased.)

\item \textbf{Configuration of type $2$}.
In this configuration, all the robots belong to $d$.blocks (that is, there are no isolated robots) and all the $d$.blocks have the same size.

If the configuration is this type, the robots neighboring to hole and with the maximum view are allowed to move to their adjacent empty nodes.
If there exists such a configuration with more than one robot and two of them may move face-to-face on the hole on the axis of symmetry, then they withdraw their candidacy and other robots with the second maximum view are allowed to move.

\item \textbf{Configuration of type $3$}.
In this configuration, the configuration is not type $1$ and $2$, \textit{i.e.}, all the other cases.
Then, there is at least one biggest $d$.block whose size is the biggest.

\begin{itemize*}
 \item If there exists an isolated robot that is neighboring to the biggest $d$.block, then it is allowed to move to the adjacent empty node towards the nearest neighboring biggest $d$.block. 
 If there exist more than one such isolated robots, then only robots that are closest to the biggest $d$.block among them are allowed to move.
 If there exist more than one such isolated robots, then only robots with the maximum view among such isolated robots are allowed to move.
The destination is their adjacent empty nodes towards one of the nearest neighboring biggest $d$.blocks.
\item If there exist no isolated robot that is neighboring to the biggest $d$.block, the robot that does not belong to the biggest $d$.block and is neighboring to the biggest $d$.block is allowed to move.
If there exists more than one such a robot, then only robots with the maximum view among them are allowed to move.
The destination is their adjacent empty node towards one of the nearest neighboring biggest $d$.blocks.
(Note that the size of the biggest $d$.block has increased.)
\end{itemize*}
\end{itemize*} 

\paragraph{Correctness of the algorithm}
In the followings, we prove the correctness of our algorithm.

\begin{lemma}\label{periodic-proof}
From any non-periodic initial configuration without tower, the algorithm does not create a periodic configuration.
\end{lemma}

\begin{proof}
Assume that, after a robot $A$ moves, the system reaches a periodic configuration $C^*$. Let $C$ be the configuration that $A$ observed to decide the movement. The important remark is that, since we assume an odd number of robots, any periodic configuration should have at least three $d$.blocks with the same size or at least three isolated robots.

\noindent\textbf{$C$ is a configuration of type 1}. Then, $C$ has a single $d$.block and there is another robot $B$ that is allowed to move. After configuration $C$, three cases are possible: $A$ moves before $B$ moves, $A$ moves after $B$ moves, or $A$ and $B$ move at the same time. After the movement of all the cases, the configuration $C^*$ has exactly one $(d-1)$.block and thus $C^*$ is not periodic.

\noindent\textbf{$C$ is a configuration of type 2}. Let $s$ be the size of $d$.blocks in $C$ (Remind that all $d$.blocks have the same size). Since the number of robots is odd, $s\ge 3$ holds.
\begin{itemize*}
\item If $C$ is not symmetric, only $A$ is allowed to move. When $A$ moves, it either becomes an isolated robot or joins another $d$.block. Then, $C^*$ has exactly one isolated robot or exactly one $d$.block with size $s+1$. Therefore, $C^*$ is not periodic.

\item If $C$ is symmetric, there is another robot $B$ that is allowed to move in $C$.
\begin{itemize*}
\item If $A$ moves before $B$ moves, $C^*$ is not periodic similarly to the non-symmetric case.
\item If $A$ and $B$ move at the same time, they become isolated robots or join other $d$.blocks. Then, three cases are possible: $C^*$ has exactly two isolated robots, $C^*$ has exactly two $d$.blocks with size $s+1$, or $C^*$ has exactly one $d$.block with size $s+2$. For all the cases $C^*$ is not periodic.

\item Consider the case that $B$ moves before $A$ moves. Then, $B$ becomes an isolated robot or joins another $d$.block. Let $C'$ be the configuration after $B$ moves. If $A$ moves in $C'$, $C^*$ is not periodic since $C^*$ is the same one as $A$ and $B$ move at the same time in $C$. Consequently, the remaining case is that other robots other than $A$ move in $C'$.

First, we consider the case that $B$ joins $d$.block in $C'$. Then, the $d$.block becomes a single biggest $d$.block. This implies that all other robots move toward this $d$.block. Consequently, the following configurations contain exactly one biggest $d$.block, and thus $C^*$ is not periodic.

Second, we consider the case that $B$ is an isolated robot in $C'$. Since only $B$ is an isolated robot in $C'$, only $B$ is allowed to move in $C'$ and it moves toward its neighboring $d$.block. If $A$ moves before $B$ joins the $d$.block, $C^*$ contains exactly two isolated robots, and thus $C^*$ is not periodic. After $B$ joins the $d$.block, $C^*$ is not periodic similarly to the previous case.
\end{itemize*}
\end{itemize*}
\noindent\textbf{$C$ is a configuration of type 3}. Let $s$ be the size of biggest $d$.blocks in $C$.
\begin{itemize*}
\item If $C$ is not symmetric, only $A$ is allowed to move. When $A$ moves, it either becomes an isolated robot or joins another $d$.block. In the latter case, $C^*$ has exactly one $d$.block with size $s+1$, and thus $C^*$ is not periodic. In the previous case, there may exist multiple isolated robots. However, $A$ is the only one isolated robot such that the distance to its neighboring biggest $d$.block is the minimum. This means $C^*$ is not periodic.
\item If $C$ is symmetric, there is another robot $B$ that is allowed to move in $C$.
\begin{itemize*}
\item If $A$ moves before $B$ moves, $C^*$ is not periodic similarly to the non-symmetric case.
\item If $A$ and $B$ move at the same time, they become isolated robots or join other $d$.blocks. In the latter case, $C^*$ has exactly two $d$.block with size $s+1$ or exactly one $d$.block with size $s+2$, and thus $C^*$ is not periodic. In the previous case, $A$ and $B$ are only two isolated robots such that the distance to its neighboring biggest $d$.block is the minimum. For both cases, $C^*$ is not periodic.
\item Consider the case that $B$ moves before $A$ moves. Then, $B$ becomes an isolated robot or joins another $d$.block. Let $C'$ be the configuration after $B$ moves. If $A$ moves in $C'$, $C^*$ is not periodic since $C^*$ is the same one as $A$ and $B$ move at the same time in $C$. Consequently, the remaining case is that other robots other than $A$ move in $C'$.

First, we consider the case that $B$ joins $d$.block in $C'$. Then, the $d$.block becomes a single biggest $d$.block. This implies that all other robots move toward this $d$.block. Consequently, the following configurations contain exactly one biggest $d$.block, and thus $C^*$ is not periodic.

Second, we consider the case that $B$ is an isolated robot in $C'$. Then, $B$ is the only one isolated robot such that the distance to its neighboring biggest $d$.block is the minimum. Consequently, only $B$ is allowed to move in $C'$ and it moves toward its neighboring $d$.block. Even if $A$ moves before $B$ joins the $d$.block, $B$ is the only one isolated robot such that the distance to its neighboring biggest $d$.block is the minimum. Consequently, in this case $C^*$ is not periodic. After $B$ joins the $d$.block, $C^*$ is not periodic similarly to the previous case.
\end{itemize*}
\end{itemize*}
For all cases, $C^*$ is not periodic; thus, a contradiction.
\end{proof}

\begin{lemma}\label{tower-proof}
No tower is created before reaching a configuration with a single 1.block for the first time.
\end{lemma}

\begin{proof}
If each robot that is allowed to move immediately moves until other robots take new snapshots, that is, no robot has outdated view, then it is clear that no tower is created.

Assume that a tower is created.
Then, two robots $A$ and $B$ were allowed to move in a configuration, but the scheduler activates only $A$, and other robot $C$ takes a snapshot after the movement of $A$ before $B$ moves.
Because $B$ moves based on the outdated view, if $B$ and $C$ moves face-to-face, then $B$ and $C$ may make a tower.

By the algorithm, in a view of a configuration, two robots are allowed to move if and only if the configuration is symmetric,
because the maximum view is only one for each configuration other than symmetric configurations.
If the configuration is not symmetric, only one robot $E$ is allowed to move and the view of each robot does not change until $E$ moves.
Therefore, we should consider only symmetric configurations, and $A$ and $B$ are two symmetric robots.

\begin{itemize*}
\item Consider the configuration of type $1$ as before $A$ moves.
Then, there is a robot $F$ on the axis of symmetry, and $F$ is not allowed to move.
The robots $A$ and $B$ are neighbor to $F$.
In the case where the scheduler activates only $A$ and $A$ moves, then $A$ and $F$ create a new $d-1$.block.
By the algorithm of the type $3$(1), the closest isolated robot to this new $d-1$.bock is allowed to move in the new view.
However, it is robot $B$, because the distance from $B$ to the $d-1$.block is $d$ but from other neighbor of the $d-1$.block is $d+1$.
Therefore, other robots cannot move, and this is a contradiction.

\item Consider the configuration of type $2$ as before $A$ moves.
Then, by the exception, the robots that are face-to-face on the hole on the axis of symmetry are not allowed to move.
Therefore, $A$ and $B$ are not neighboring to such hole on the axis of symmetry.
After $A$ moves, $A$ becomes isolated or joins the other $d$.block, and the $d$.block $A$ belonged to becomes not the biggest on the new view.
If $A$ becomes isolated, by the algorithm of type $3$(1), $A$ are allowed to move on the new view.
Therefore, $A$ can move and others than $B$ cannot move until it joins the neighboring biggest $d$.block.
After $A$ joins the other $d$.block, the configuration becomes type $3$(2) and the $d$.block $D$ $A$ belongs to is the biggest while $B$ does not move.
After that, other processes $C$ neighboring to $D$ can move toward $D$.
Because $B$ is not neighboring to $D$, $C$ and $B$ cannot move face-to-face.
This is a contradiction.

\item Consider the configuration of type $3$(1) as before $A$ moves.
Then, $A$ and $B$ are isolated robots that are neighboring to the biggest $d$.block.
Their destinations are the empty nodes towards the nearest neighboring biggest $d$.blocks respectively.
The configuration after $A$ move is type $3$(1) until $A$ join the nearest neighboring biggest $d$.block $D$.
Then, $A$ and $B$ are allowed to move and others cannot move until $A$ joins $D$.
Consider the case after $A$ joins $D$.
\begin{itemize*}
\item If there exist other isolated robot $C$ neighboring to $D$, then it can move toward $D$ because the configuration becomes type $3$(1).
However, $C$ and $B$ cannot move face-to-face because $C$ are neighboring to $D$ and $B$ moves toward the border of other $d$.block.
This is a contradiction.
\item If there does not exist other isolated robot neighboring to $D$, then the configuration becomes type $3$(2) on the new view.
Then, the robot $C$ neighboring to $D$ can move toward $D$.
However, $C$ and $B$ cannot move face-to-face because $C$ are neighboring to $D$ and $B$ moves toward the other neighboring $d$.block.
This is a contradiction.
\end{itemize*}

\item Consider the configuration of type $3$(2) as before $A$ moves.
Then, $A$ and $B$ are neighboring to the biggest $d$.blocks and are members of any (not biggest) $d$.blocks.
After $A$ moves, $A$ becomes isolated until $A$ joins the biggest $d$.block $D$ that is the destination of $A$.
By the algorithm of type $3$(1), $A$ is allowed to move on the new view and others cannot move.
After $A$ joins $D$, $D$ becomes biggest, the configuration becomes type $3$(2).
Then, the other process $C$ neighboring to $D$ can move to $D$.
However, $B$ and $C$ cannot move face-to-face because $C$ are neighboring to $D$ and $B$ moves toward the other neighboring $d$.block.
This is a contradiction.                                                                         
\end{itemize*}
From the cases above, we can deduct that no tower is created before the gathering process.
\end{proof}

From Lemmas \ref{periodic-proof} and \ref{tower-proof}, the configuration is always non-periodic and does not have a tower from any non-periodic initial configuration without tower. 
Since configurations are not periodic, there exist one or two robots that are allowed to move unless the configuration contains a single $1$.block.

\begin{lemma}\label{lem-time1}
Let $C$ be a configuration such that its inter-distance is $d$ and the size of the biggest $d$.block is $s$ ($s \le k-1$). From configuration $C$, the configuration becomes such that the size of the biggest $d$.block is at least $s+1$ in $O(n)$ rounds.
\end{lemma}

\begin{proof}
From configurations of type 2 and type 3, at least one robot neighboring to the biggest $d$.block is allowed to move. Consequently, the robot moves in $O(1)$ rounds. If the robot joins the biggest $d$.block, the lemma holds.

If the robot becomes an isolated robot, the robot is allowed to move toward the biggest $d$.block by the configurations of type 3 (1). Consequently the robot joins the biggest $d$.block in $O(n)$ rounds, and thus the lemma holds.
\end{proof}

\begin{lemma}\label{lem-time2}
Let $C$ be a configuration such that its inter-distance is $d$. From configuration $C$, the configuration becomes such that there is only single $d$.block in $O(kn)$ rounds.
\end{lemma}

\begin{proof}
From Lemma \ref{lem-time1}, the size of the biggest $d$.block becomes larger in $O(n)$ rounds. Thus, the size of the biggest $d$.block becomes $k$ in $O(kn)$ rounds. Since the configuration that has a $d$.block with size $k$ is the one such that there is only single $d$.block. Therefore, the lemma holds.
\end{proof}

\begin{lemma}\label{lem-time3}
Let $C$ be a configuration such that there is only single $d$.block ($d\ge 2$). From configuration $C$, the configuration becomes one such that there is only single $(d-1)$.block in $O(kn)$ rounds.
\end{lemma}

\begin{proof}
From the configuration of type 1, the configuration becomes one such that there is $(d-1)$.block in $O(1)$ rounds. After that, the configuration becomes one such that there is only single $(d-1)$.block in $O(kn)$ rounds by Lemma \ref{lem-time2}. Therefore, the lemma holds.
\end{proof}

\begin{lemma}\label{lem-time}
From any non-periodic initial configuration without tower, the configuration becomes one such that there is only single $1$.block in $O(n^2)$ rounds.
\end{lemma}

\begin{proof}
Let $d$ be the inter-distance of the initial configuration. From the initial configuration, the configuration becomes one such that there is a single $d$.block in $O(kn)$ rounds by Lemma \ref{lem-time2}. Since the inter-distance becomes smaller in $O(kn)$ rounds by Lemma \ref{lem-time3}, the configuration becomes one such that there is only single $1$.block in $O(dkn)$ rounds. Since $d\le n/k$ holds, the lemma holds.
\end{proof}

\subsection{Second phase: An algorithm to achieve the gathering}

In this section, we provide the algorithm for the second phase, that is, the algorithm to achieve the gathering from any configuration with a single $1$.block. As described in the beginning of this section, to separate the behavior from the one to construct a single $1$.block, we define a special configuration set $\CONF_{sp}$ that includes any configuration with a single $1$.block. Our algorithm guarantees that the system achieves the gathering from any configuration in $\CONF_{sp}$ without moving out of $\CONF_{sp}$. We combine two algorithms for the first phase and the second phase in the following way: Each robot executes the algorithm for the second phase if the current configuration is in $\CONF_{sp}$, and executes one for the first phase otherwise. By this way, as soon as the system becomes a configuration in $\CONF_{sp}$ during the first phase, the system moves to the second phase and the gathering is eventually achieved. Note that the system moves to the second phase without creating a single $1$.block if it reaches a configuration in $\CONF_{sp}$ before creating a single $1$.block.


The strategy of the second phase is as follows. 
When a configuration with a single 1.block is reached, the configuration becomes symmetric. Note that since there is an odd number of robots, we are sure that there is one robot $R1$ that is on the axis of symmetry. 
The two robots that are neighbor of $R1$ move towards $R1$. 
Thus $R1$ will have two neighboring holes of size $1$. 
The robots that are neighbor of such a hole not being on the axis of symmetry move towards the hole. 
By repeating this process, a new $1$.block is created (Note that its size has decreased and the tower is on the axis of symmetry). 
Consequently robots can repeat the behavior and achieve the gathering. 
Note that due to the asynchrony of the system, the configuration may contain a single $1$.block of size $2$. 
In this case, one of the two nodes of the block contains a tower (the other is occupied by a single robot). 
Since we assume a local weak multiplicity detection, only the robot that does not belong to a tower can move. 
Thus, the system can achieve the gathering. 

In the followings, we define the special configuration set $\CONF_{sp}$ and the behavior of robots in the configurations. To simplify the explanation, we define a block as a maximal consecutive nodes where every node is occupied by some robots. The size $Size(B)$ of a block $B$ denotes the number of nodes in the block. Then, we regard an isolated node as a block of size 1. 

The configuration set $\CONF_{sp}$ is partitioned into five subsets: Single block $\CONF_{sb}$, block leader $\CONF_{bl}$, semi-single block $\CONF_{ssb}$, semi-twin $\CONF_{st}$, semi-block leader $\CONF_{sbl}$. That is, $\CONF_{sp}=\CONF_{sb}\cup\CONF_{bl}\cup\CONF_{ssb}\cup\CONF_{st}\cup\CONF_{sbl}$ holds. We provide the definition of each set and the behavior of robots. Note that, although the definition of configurations specifies the position of a tower, each robot can recognize the configuration without detecting the position of a tower if the configuration is in $\CONF_{sp}$.
\begin{itemize*}
\item {\bf Single block.} A configuration $C$ is a single block configuration (denoted by $C\in\CONF_{sb}$) if and only if there exists exactly one block $B0$ such that $Size(B0)$ is odd or equal to $2$. Note that If $Size(B0)$ is equal to 2, one node of $B0$ is a tower and the other node is occupied by one robot. If $Size(B0)$ is odd, letting $v_t$ be the center node of $B0$, no node other than $v_t$ is a tower.


In this configuration, robots move as follows: 1) If $Size(B0)$ is equal to 2, the robot that is not on a tower moves to the neighboring tower. 2) If $Size(B0)$ is odd, the configuration is symmetric and hence there exists one robot on the axis of symmetry (Let this robot be $R1$). Then, the robots that are neighbors of $R1$ move towards $R1$.

\item {\bf Block leader.} A configuration $C$ is a block leader configuration (denoted by $C\in\CONF_{bl}$) if and only if the following conditions hold (see Figure \ref{BL}): 1) There exist exactly three blocks $B0$, $B1$, and $B2$ such that $Size(B0)$ is odd and $Size(B1)=Size(B2)$. 2) Blocks $B0$ and $B1$ share a hole of size 1 as their neighbors. 3) Blocks $B0$ and $B2$ share a hole of size 1 as their neighbors. 4) Letting $v_t$ be the center node in $B0$, no node other than $v_t$ is a tower. Note that, since $k<n-3$ implies that there exist at least four free nodes, robots can recognize $B0$, $B1$, and $B2$ exactly.

In this configuration, the robots in $B1$ and $B2$ that share a hole with $B_0$ as its neighbor move towards $B0$.

\item {\bf Semi-single block.} A configuration $C$ is a semi-single block configuration (denoted by $C\in\CONF_{ssb}$) if and only if the following conditions hold (see Figure \ref{SSB}): 1) There exist exactly two blocks $B1$ and $B2$ such that $Size(B2)=1$ and $Size(B1)$ is even (Note that this implies $Size(B1)+Size(B2)$ is odd.). 2) Blocks $B1$ and $B2$ share a hole of size 1 as their neighbors. 3) Letting $v_t$ be a node in $B1$ that is the $(Size(B1)/2)$-th node from the side sharing a hole with $B2$, no node other than $v_t$ is a tower.

In this configuration, the robot in $B2$ moves towards $B1$.

\item {\bf Semi-twin.} A configuration $C$ is a semi-twin configuration (denoted by $C\in\CONF_{st}$) if and only if the following conditions hold (see Figure \ref{ST}). 1) There exist exactly two blocks $B1$ and $B2$ such that $Size(B2)=Size(B1)+2$ (Note that this implies $Size(B1)+Size(B2)$ is even, which is distinguishable from semi-single block configurations). 2) Blocks $B1$ and $B2$ share a hole of size 1 as their neighbors. 3) Letting $v_t$ be a node in $B2$ that is the neighbor of a hole shared by $B1$ and $B2$, no node other than $v_t$ is a tower.

In this configuration, the robot in $B2$ that is a neighbor of $v_t$ moves towards $v_t$.

\item {\bf Semi-block leader.} A configuration $C$ is a semi-block leader configuration (denoted by $C\in\CONF_{sbl}$) if and only if the following conditions hold (see Figure \ref{SBL}). 1) There exist exactly three blocks $B0$, $B1$, and $B2$ such that $Size(B0)$ is even and $Size(B2)=Size(B1)+1$. 2) Blocks $B0$ and $B1$ share a hole of size 1 as their neighbors. 3) Blocks $B0$ and $B2$ share a hole of size 1 as their neighbors. 4) Letting $v_t$ be a node in $B0$ that is the $(Size(B0)/2)$-th node from the side sharing a hole with $B2$, no node other than $v_t$ is a tower. Note that, since $k<n-3$ implies that there exist at least four free nodes, robots can recognize $B0$, $B1$, and $B2$ exactly.

In this configuration, the robot in $B2$ that shares a hole with $B0$ as a neighbor moves towards $B0$.
\end{itemize*}

\begin{figure}[t!]
 \begin{minipage}[b]{.3\linewidth}
  \centering\epsfig{figure=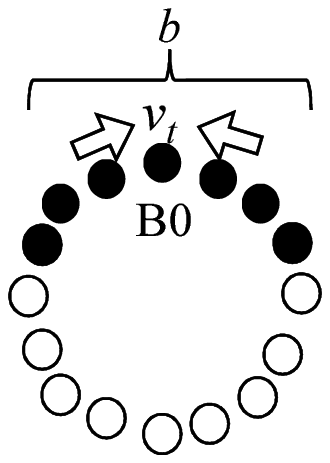,scale=0.80}
  \caption{Single block $\CONF_{sb}(b)$ ($b>2$)\label{SB}}
 \end{minipage} \hfill
 \begin{minipage}[b]{.3\linewidth}
  \centering\epsfig{figure=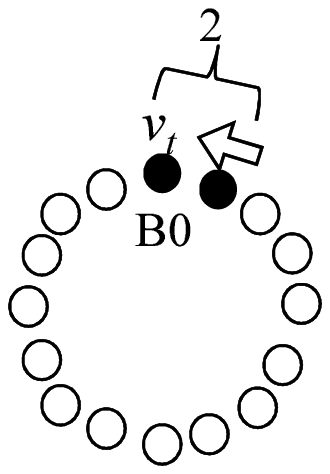,scale=0.80}
  \caption{Single block $\CONF_{sb}(2)$\label{SB2}}
 \end{minipage} \hfill
 \begin{minipage}[b]{.3\linewidth}
  \centering\epsfig{figure=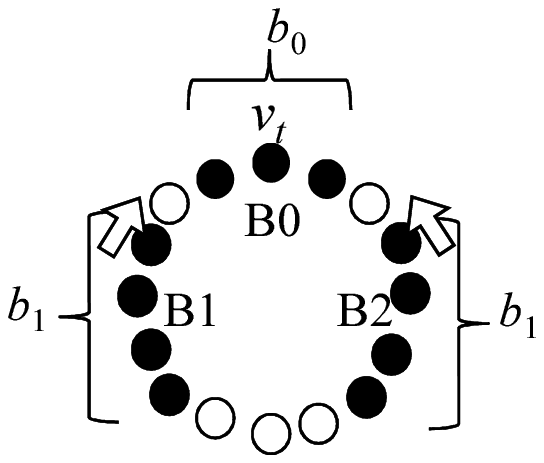,scale=0.80}
  \caption{Block leader $\CONF_{bl}(b_0,b_1)$\label{BL}}
 \end{minipage} \hfill
 \begin{minipage}[b]{.3\linewidth}
  \centering\epsfig{figure=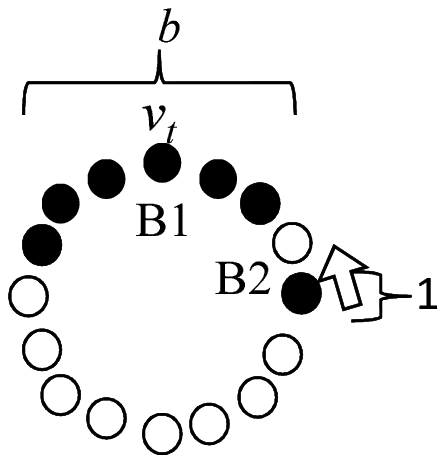,scale=0.80}
  \caption{Semi-single block $\CONF_{ssb}(b)$\label{SSB}}
 \end{minipage} \hfill
 \begin{minipage}[b]{.3\linewidth}
  \centering\epsfig{figure=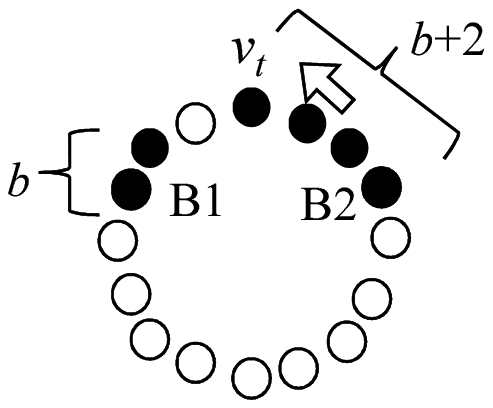,scale=0.80}
  \caption{Semi-twin $\CONF_{st}(b)$\label{ST}}
 \end{minipage} \hfill
 \begin{minipage}[b]{.3\linewidth}
  \centering\epsfig{figure=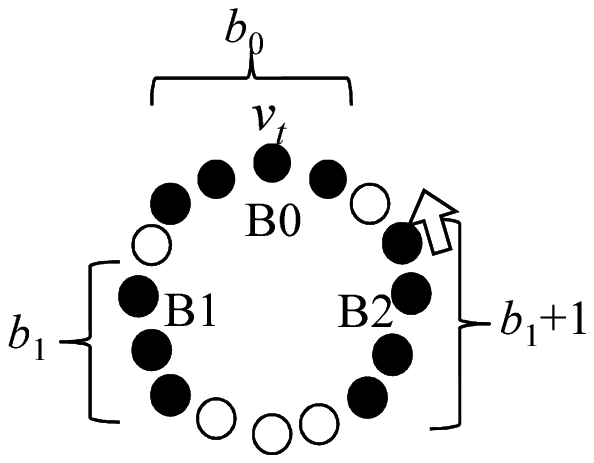,scale=0.80}
  \caption{Semi-block leader $\CONF_{sbl}(b_0,b_1)$\label{SBL}}
 \end{minipage} \hfill
\end{figure}

\paragraph{Correctness of the algorithm}
In the followings, we prove the correctness of our algorithm. To prove the correctness, we define the following notations.
\begin{itemize*}
\item $\CONF_{sb}(b)$: A set of single block configurations such that $Size(B0)=b$.
\item $\CONF_{bl}(b_0,b_1)$: a set of block leader configurations such that $Size(B0)=b_0$ and $Size(B1)=Size(B2)=b_1$.
\item $\CONF_{ssb}(b)$: A set of semi-single block configurations such that $Size(B1)=b$.
\item $\CONF_{st}(b)$: A set of semi-twin configurations such that $Size(B1)=b$.
\item $\CONF_{sbl}(b_0,b_1)$: A set of semi-block leader configurations such that $Size(B0)=b_0$ and $Size(B1)=b_1$.
\end{itemize*}
Note that every configuration in $\CONF_{sp}$ has at most one node that can be a tower (denoted by $v_t$ in the definition). We denote such a node by a tower-construction node. In addition, we define an outdated robot as the robot that observes the outdated configuration and tries to move based on the outdated configuration. 

From Lemmas \ref{lem:sbtobl} to \ref{lem:sptogat}, we show that, from any configuration $C\in\CONF_{sp}$ with no outdated robots, the system achieves the gathering. However, some robots may move based on the configuration in the first phase. That is, some robots may observe the configuration in the first phase and try to move, however the configuration reaches one in the second phase before they move. Thus, in Lemma \ref{twoalg}, we show that the system also achieves the gathering from such configurations with outdated robots.

\begin{lemma}
\label{lem:sbtobl}
From any single block configuration $C\in\CONF_{sb}(b)$ ($b\geq 5$) with no outdated robots, the system reaches a configuration $C'\in\CONF_{bl}(1,(b-3)/2)$ with no outdated robots in $O(1)$ rounds.
\end{lemma}

\begin{proof}
In $C$ the robots that are neighbors of the tower-construction node can move. Two sub cases are possible. First, the scheduler makes the two robots move at the same time. Once the robots move, they join the tower-construction node. Then, the configuration becomes one in $\CONF_{bl}(1,(b-3)/2)$ and there exist no outdated robots.
In the second case, the scheduler activates the two robots separately. In this case, one of the two robots first joins the tower-construction node and the configuration becomes one in $\CONF_{st}((b-3)/2)$. After that, the other robot joins the tower-construction node (No other robots can move in this configuration). Then, the configuration becomes one in $\CONF_{bl}(1,(b-3)/2)$ and there exist no outdated robots.
In both cases, the transition requires at most $O(1)$ rounds and thus the lemma holds.
\end{proof}

\begin{lemma}
\label{lem:sbtoga}
From any single block configuration $C\in\CONF_{sb}(3)$ with no outdated robots, the system achieves the gathering in $O(1)$ rounds.
\end{lemma}

\begin{proof}
In $C$ the robots that are neighbors of the tower-construction node can move. Two sub cases are possible.
First, the scheduler makes the two robots move at the same time. Once the robots move, they join the tower-construction node. Then, the system achieves the gathering.
In the second case, the scheduler activates the two robots separately. In this case, one of the two robots first joins the tower-construction node and the configuration becomes one in $\CONF_{sb}(2)$. After that, the other robot joins the tower-construction node because robots on the tower never move due to the local multiplicity detection. And thus, the system achieves the gathering.
In both cases, the transition requires at most $O(1)$ rounds and thus the lemma holds.
\end{proof}

\begin{lemma}
\label{lem:bltobl}
From any block leader configuration $C\in\CONF_{bl}(b_0,b_1)$ ($b_1\geq 2$) with no outdated robots, the system reaches a configuration $C'\in\CONF_{bl}(b_0+2,b_1-1)$ with no outdated robots in $O(1)$ rounds.
\end{lemma}

\begin{proof}
In $C$ the robots in $B1$ and $B2$ that share a hole of size 1 with $B0$ can move. Two sub cases are possible.
First, the scheduler makes the two robots move at the same time. Once the robots move, they join $B0$. Since the size of $B0$ is increased by two and the size of $B1$ and $B2$ is decreased by one, the configuration becomes one in $\CONF_{bl}(b_0+2,b_1-1)$. In addition, there exist no outdated robots.
The second possibility is that the scheduler activates the two robots separately. In this case, one of the two robots first joins $B0$. Then, since the size of $B0$ is increased by one and the size of either $B1$ or $B2$ is decreased by one, the configuration becomes one in $\CONF_{sbl}(b_0+1,b_1-1)$. After that, the other robot joins $B0$ (No other robots can move in this configuration). Then, the configuration becomes one in $\CONF_{bl}(b_0+2,b_1-1)$ and there exist no outdated robots.
In both cases, the transition requires at most $O(1)$ rounds and thus the lemma holds.
\end{proof}

\begin{lemma}
\label{lem:bltosb}
From any block leader configuration $C\in\CONF_{bl}(b_0,1)$ with no outdated robots, the system reaches a configuration $C'\in\CONF_{sb}(b_0+2)$ with no outdated robots in $O(1)$ rounds.
\end{lemma}

\begin{proof}
In $C$ the robots in $B1$ and $B2$ can move. Two sub cases are possible: \emph{(i)} The scheduler makes the two robots move at the same time, then once the robots move, they join $B0$. Then, the system reaches a configuration in $\CONF_{sb}(b_0+2)$ with no outdated robots. \emph{(ii)} The scheduler activates the two robots separately. In this case, one of the two robots first joins $B0$. Then, the configuration becomes one in $\CONF_{ssb}(b_0+1)$. After that, the other robot joins $B0$, and thus the configuration becomes one in $\CONF_{sb}(b_0+2)$ with no outdated robots.
In both cases, the transition requires at most $O(1)$ rounds and thus the lemma holds.
\end{proof}

\begin{lemma}
\label{lem:sbtosb}
From any single block configuration $C\in\CONF_{sb}(b)$ ($b\geq 5$) with no outdated robots, the system reaches a configuration $C'\in\CONF_{sb}(b-2)$ with no outdated robots in $O(k)$ rounds.
\end{lemma}

\begin{proof}
From Lemma \ref{lem:sbtobl}, the configuration becomes one in $\CONF_{bl}(1,(b-3)/2)$ in $O(1)$ rounds. After that, from Lemmas \ref{lem:bltobl} and \ref{lem:bltosb}, the configuration becomes $\CONF_{sb}(b-2)$ in at most $O((b-3)/2)$ rounds. Since $b\leq k$, the lemma holds.
\end{proof}

\begin{lemma}
\label{lem:sbtogat}
From any single block configuration $C\in\CONF_{sb}(b)$ with no outdated robots, the system achieves the gathering in $O(k^2)$ rounds.
\end{lemma}

\begin{proof}
From Lemma~\ref{lem:sbtosb}, if $b\geq 5$, the size of the block is decreased by two in $O(k)$ rounds. From Lemma~\ref{lem:sbtoga}, if the size of the block is 3, the system achieves the gathering in $O(1)$ rounds. 
\end{proof}

Lemma~\ref{lem:sbtogat} says that the system achieves the gathering in $O(k^2)$ rounds from any single block configuration in $\CONF_{sb}$. For other configurations, we can show the following lemmas.

\begin{lemma}
\label{lem:bltogat}
From any block leader configuration $C \in \CONF_{bl}$ with no outdated robots, the system achieves the gathering in $O(k^2)$ rounds.
\end{lemma}

\begin{proof}
Consider configuration $C\in\CONF_{bl}(b_0,b_1)$. From Lemmas~\ref{lem:bltobl} and ~\ref{lem:bltosb}, the configuration becomes $\CONF_{sb}(b_0+2b_1)$ in $O(k)$ rounds since $b_1\le k$ holds. After that, from Lemma~\ref{lem:sbtogat}, the system achieves the gathering in $O(k^2)$ rounds.
\end{proof}

\begin{lemma}
\label{lem:ssbtogat}
From any semi-single block configuration $C\in\CONF_{ssb}$ with no outdated robots, the system achieves the gathering in $O(k^2)$ rounds.
\end{lemma}

\begin{proof}
Consider configuration $C\in\CONF_{ssb}(b)$. Then, the configuration becomes $\CONF_{sb}(b+1)$ in $O(1)$ rounds. After that, from Lemma \ref{lem:sbtogat}, the system achieves the gathering in $O(k^2)$ rounds.
\end{proof}

\begin{lemma}
\label{lem:sttogat}
From any semi-twin configuration $C\in\CONF_{st}$ with no outdated robots, the system achieves the gathering in $O(k^2)$ rounds.
\end{lemma}

\begin{proof}
Consider configuration $C\in\CONF_{st}(b)$. Then, the configuration becomes $\CONF_{bl}(1,b)$ in $O(1)$ rounds. After that, from Lemma \ref{lem:bltogat}, the system achieves the gathering in $O(k^2)$ rounds.
\end{proof}

\begin{lemma}
\label{lem:sbltogat}
From any semi-block leader configuration $C\in\CONF_{sbl}$ with no outdated robots, the system achieves the gathering in $O(k^2)$ rounds.
\end{lemma}

\begin{proof}
Consider configuration $C\in\CONF_{sbl}(b_0,b_1)$.
Then, the configuration becomes $\CONF_{bl}(b_0+1,b_1)$ in $O(1)$ rounds.
After that, from Lemma~\ref{lem:bltogat}, the system achieves the gathering in $O(k^2)$ rounds.
\end{proof}

From Lemmas~\ref{lem:sbtogat},~\ref{lem:bltogat}, ~\ref{lem:ssbtogat}, ~\ref{lem:sttogat} and ~\ref{lem:sbltogat}, we have the following lemma.

\begin{lemma}
\label{lem:sptogat}
From any configuration $C\in\CONF_{sp}$ with no outdated robots, the system achieves the gathering in $O(k^2)$ rounds.
\end{lemma}

\begin{lemma}
\label{twoalg}
From any configuration $C\in\CONF_{sp}$ with outdated robots, the system achieves the gathering in $O(k^2)$ rounds.
\end{lemma}

\begin{proof}
We call the algorithm to construct a single $1$.block {\sf Alg1}, and
the algorithm to achieve the gathering {\sf Alg2}.

To construct a configuration $C\in\CONF_{sp}$ with outdated robots during {\sf Alg1}, two robots $P$ and $Q$ have to observe a symmetric configuration $C^*$ if they are activated by the scheduler they will move. 
From the behavior of {\sf Alg1}, $P$ and $Q$ move toward their neighboring biggest $d$.blocks. Since $P$ and $Q$ observe a symmetric configuration, the directions of their movements are different each other. 

We assume that $P$ was activated by the scheduler however it executed only the look phase thus it doesn't move based on a configuration in {\sf Alg1},
and the system reaches a configuration $C\in\CONF_{sp}$ by the behavior of $Q$ (and other robots that move after $Q$ joins the biggest $d$.block). Note that $P$ and $Q$ are isolated robots or the border of a block in $C^*$. 

We have two possible types of configurations as $C$ based on the behavior of $Q$.
\begin{itemize*}
\item We say $C$ is of {\sf TypeA} if $Q$ joins its neighboring biggest $d$.block between $C^*$ and $C$. In this case, the join of $Q$ creates exactly one biggest $d$.block. From the behavior of {\sf Alg1}, other robots move toward this biggest $d$.block and thus there exists only one biggest $d$.block in $C$. In addition, since $Q$ and other robots do not move to the block of $P$ from the behavior of {\sf Alg1}, the biggest $d$.block in $C$ does not include $P$.
\item We say $C$ is of {\sf TypeB} if $Q$ does not join its neighboring biggest $d$.block between $C^*$ and $C$. In this case, $Q$ is an isolated robot in $C$. In addition, positions of non-isolated robots other than $P$ and $Q$ are the same in $C$ and $C^*$.
\end{itemize*}






We consider five cases according to the type of the configuration $C$: single block, block leader, semi-single block, semi-twin, and semi-block leader.

\paragraph{Single block}
Consider the case that $C$ is single block. However, if $C$ is of {\sf TypeA}, there exist at least two blocks. If $C$ is of {\sf TypeB}, there exist at least one isolated robot. Both cases contradict the single block configuration.

\paragraph{Semi-twin}
Consider the case that $C$ is semi-twin. Since the number of nodes occupied by robots is even, there exists a tower. However, since the behavior of {\sf Alg1} does not make a tower, this is a contradiction.


\paragraph{Semi-single block} 
Consider the case that $C$ is semi-single block.

If $C$ is of {\sf TypeA}, since $P$ is not in the biggest block, $P$ is in $B2$. If $P$ moves to $B1$ by the same direction of {\sf Alg2}, the system reaches a single block configuration. If $P$ moves to the opposite direction, the system reaches a configuration with a single biggest block and one isolated robot $P$. In this configuration, only $P$ can move by {\sf Alg1} and the system reaches a semi-single block configuration with no outdated robots. Therefore, the system reaches a single block configuration by {\sf Alg2} and achieves the gathering in $O(k^2)$ rounds. 

If $C$ is of {\sf TypeB}, $Q$ is in $B2$. However, this means $Q$ moved out from $B1$ and thus the configuration was single block before $Q$ moves. Thus, this is a contradiction.


\paragraph{Block leader}
Consider the case that $C$ is block leader.

First, we assume $b_0>b_1$, that is, $B0$ is the biggest block in $C$. 
Without loss of generality, we assume $P$ is in $B1$ because the size of $B1$ and $B2$ is same.
\begin{itemize*}
\item Consider the case that $C$ is of {\sf TypeA}.
\begin{itemize*}
\item 
Consider the case that the size of $B1$ is bigger than $1$, then the size of $B2$ is also bigger than $1$.
If $P$ is the border of $B1$ that shares a hole with $B0$, the destinations of $P$ by {\sf Alg1} and {\sf Alg2} are the same. 
If $P$ is the other border of $B1$, the biggest $d$.block of the destination in $C^*$ is $B2$.
Because $2<k<n-3$, the size of hole $H$ between $P$ and $B2$ is more than two. 
Because $C^*$ is symmetric and $B2$ does not move between $C^*$ and $C$ and the size of holes other than $H$ is one, it is a contradiction.
\item 
Consider the case that the size of $B1$ is equal to $1$, the the size of $B2$ is also equal to $1$.
Then the destination of $B2$ is $B0$ by {\sf Alg2}.
If the destination of $P$ is $B0$, the system achieves the semi-single block or single block with no outdated robots.
If the destination of $P$ is $B2$, 
the system reaches a configuration with a single biggest block and at least one isolated robot $P$ with no outdated robots.
However, by {\sf Alg1}, the robot on the biggest block cannot move, 
and the isolated robots are the one that move to the biggest block.
Therefore, the system reaches a semi-single block configuration or single block configuration with no outdated robots and achieves the gathering in $O(k^2)$ rounds.
\end{itemize*}
\item Consider the case that $C$ is of {\sf TypeB}.
Because $Q$ is an isolated robot, $Q$ is part of $B2$ which is of size $1$.
Thus, $B1$ is also of size $1$.
\end{itemize*}  


Second, we assume $b_0<b_1$. Then, there exist two biggest $d$.blocks, and thus $C$ is of {\sf TypeB}. (Note that, in {\sf TypeA}, the biggest $d$.block is only one.) Therefore, $Q$ is an isolated robot and in $B0$. Without loss of generality, we assume that the destination of $Q$ is $B2$ because the size of $B1$ and $B2$ is same. Since the size of a hole between $B1$ and $Q$ is one in $C$, $Q$ belongs to $B1$ in $C^*$. This implies that the size of $B1$ is $b_1+1$ and the size of $B2$ is $b_1$ in $C^*$. Since $Q$ moves to the smaller block in $C^*$, this is a contradiction.


Finally, we assume $b_0=b_1$.
There are three biggest $d$.blocks.
If $C$ is of {\sf TypeA}, because the biggest $d$.block is only one to which $Q$ belongs.
This is a contradiction.
If $C$ is of {\sf TypeB}, because $Q$ is an isolated robot, each size of $B0$, $B1$ and $B2$ is equal to $1$.
Then, because there are only three robots and $C^*$ is symmetric, $P$ and $Q$ are not $B0$ and the destination of both of them is $B0$ by type 1 of {\sf Alg1}. 
This destination is same as by {\sf Alg2}.
%


\paragraph{Semi-block leader}
Consider the case that $C$ is semi-block leader.
Then, $B0$ or $B2$ is the biggest block in $C$.

If $b_0=b_1+1$, there are two biggest blocks. 
This implies $C$ is of {\sf TypeB}, and then $Q$ is an isolated robot.
Consequently, $B1$ contains only $Q$, and thus $B0$ and $B2$ contain two robots. 
Since $Q$ moves to the biggest block in $C^*$, $Q$ is an isolated robot in $C^*$ (Otherwise, $Q$ is in a block $B0$ or $B2$ with the size three). Remind that $P$ and $Q$ are symmetric in $C^*$.
However, $P$ is in block $B0$ or $B2$. 
This is a contradiction, and thus this case never happens.

If $b_0>b_1+1$, $B0$ is the biggest block in $C$. 
If $C$ is of {\sf TypeA}, $Q$ joins $B0$. 
Then, $P$ is in $B1$ or $B2$.
By the definition of $B2$, the size of $B2$ is bigger than $1$.
\begin{itemize*}
\item Consider the case that $P$ is in $B1$.
If the destination of $P$ in {\sf Alg1} is to $B0$, $Q$ joins from $B2$ because $C^*$ is symmetric and both destination of $P$ and $Q$ is $B0$.
However, because the size of $B2$ is bigger than $B1$, it is a contradiction.
If the destination of $Q$ in {\sf Alg1} is to $B2$, then the position of $B2$ does not change from $C^*$ because the size of $B2$ is bigger than $1$.
However, because the size of hole between $B1$ and $B2$ is more than $2$ and $C^*$ is symmetric, then there is another hole which size is more than $2$, it is a contradiction.

\item Consider the case that $P$ is in $B2$.
If the destination of $P$ in {\sf Alg1} is to $B0$, then it is same as the destination by {\sf Alg2} and the other robot does not move. Therefore, the system achieves the gathering in $O(k^2)$ rounds.
If the destination of $P$ in {\sf Alg1} is to $B1$ and the size of $B1$ is bigger than $1$, then
the position of $B1$ does not change from $C^*$.
However, because the size of hole between $B1$ and $B2$ is more than $2$ and $C^*$ is symmetric, then there is another hole which size is more than $2$, it is a contradiction.
If the destination of $P$ in {\sf Alg1} is to $B1$ and the size of $B1$ is equal to $1$, then $B1$ moves to $B0$, the size of $B2$ is $2$, and the robot $R$ in $B2$ other than $P$ moves to $B0$. Then, the members of $B0$ cannot move. If $P$ moves, after $R$ joins $B0$, then new destination of $P$ is to $B0$ by both of {\sf Alg1} and {\sf Alg2}.  Therefore, the system achieves the gathering in $O(k^2)$ rounds.
\end{itemize*}

If $C$ is of {\sf TypeB}, $Q$ is an isolated robot and thus $B1$ contains only $Q$. Consequently, $B2$ contains two robots. Since $Q$ moves toward the biggest block, $Q$ moves toward $B0$. This means $P$ is in $B2$ and moves toward $B0$. However, since $P$ has a hole of size one in its direction, $Q$ should also have a hole of size one in its direction in $C^*$. Then, $Q$ joins a block immediately after $Q$ moves, and thus $C$ never becomes of {\sf TypeB}. Therefore, this case never happens.

If $b_0<b_1+1$, then $B2$ is the biggest $d$.block. 
If $C$ is of {\sf TypeA}, $Q$ joins $B2$. Then, $P$ is in $B0$ or $B1$. 

\begin{itemize*}
\item Consider the case that $P$ is in $B0$. 
If the destination of $P$ is $B1$, then $Q$ joins $B2$ from $B0$ because the size of a hole between $B0$ and $B2$ is one.
However, in $C^*$, the configuration is block leader, it is a contradiction.
If the destination of $P$ is $B2$, then $Q$ joins $B2$ from the side of $B1$.
If $Q$ is a member of $B1$ in $C^*$, then the size of $B1$ is bigger than $B2$.
If $Q$ is not a member of $B1$ in $C^*$, then until all block members to which $Q$ belongs in $C^*$ join $B2$, the size of hole between the block and $B2$ is $1$ because the size of hole between $P$ and $B2$ is $1$.
Therefore, in $C$, all members to which $Q$ belongs in $C^*$ join $B2$.
By considering the size of $B1$ and $B2$, $Q$ is an isolated robot in $C^*$.
Therefore, $P$ in $B0$ is an isolated robot, that is, the number of robots in this case is $4$.
This is a contradiction.   

\item Consider the case that $P$ is in $B1$. 
\begin{itemize*}
\item Consider the case that $P$ is a neighbor to $B2$. Then, in $C$, $Q$ moves to $B2$, so $Q$ is in $B0$. Because both size of holes between $B1$ and $B0$ and between $B0$ and $B2$ is $1$, and $P$ and $Q$ are symmetric in $C^*$, the size of a hole between $B1$ and $B2$ is also $1$. Because $n\geq k+4$, it is a contradiction.
\item Consider the case that $P$ is a neighbor to $B0$. Then, $P$ tries to move toward $B0$. Since $P$ and $Q$ are symmetric in $C^*$, $Q$ joins $B2$ from the side of $B1$.

Remind that the difference between the size of $B1$ and that of $B2$ is one in $C$ and the size of $B1$ does not increase from $C^*$ to $C$. Therefore, the size of $B1$ and that of $B2$ are the same immediately before $Q$ joins $B2$. On the other hand, since $P$ and $Q$ are symmetric in $C^*$ and $P$ belongs to $B1$, $Q$ also belongs to $B1$. This implies the size of $B1$ is bigger than that of $B2$ in $C^*$. This is a contradiction.
\end{itemize*}
\end{itemize*}

If $C$ is of {\sf TypeB}, $Q$ is an isolated robot.

\begin{itemize*}
\item Consider the case that $P$ is in $B0$. Then, the member of $B1$ is only $Q$, and the size of $B2$ is $b_1+1=2$. Because $b_0<b_1+1=2$, the member of $B0$ is only $P$. Then, the number of robot is $4$, and this is a contradiction.
\item Consider the case that $P$ is in $B1$. Then, the member of $B0$ is only $Q$. Since $Q$ moves toward the biggest block, $Q$ moves toward $B2$. This implies $Q$ is in $B1$ in $C^*$, and consequently there exist only two blocks with the same size in $C^*$. Since $P$ and $Q$ are symmetric in $C^*$, the sizes of holes in both directions are the same. This implies $C^*$ is periodic, and this is a contradiction.
\item Consider the case that $P$ is in $B2$. Then, since $P$ is in the biggest block in $C$, $P$ is also in the biggest block in $C^*$. Consequently, $P$ moves toward the neighboring biggest block which is the same size as $B2$. This implies the size of $B0$ is at least $b_1+1$, and this is a contradiction.
\end{itemize*}

\end{proof}

From Lemmas \ref{lem-time}, \ref{lem:sptogat} and \ref{twoalg}, we have the following theorem.

\begin{theorem}
From any non-periodic initial configuration without tower, the system achieves the gathering in $O(n^2)$ rounds.
\end{theorem}

\section{Concluding remarks}

We presented a new protocol for mobile robot gathering on a ring-shaped network. Contrary to previous approaches, our solution neither assumes that global multiplicity detection is available nor that the network is started from a non-symmetric initial configuration. Nevertheless, we retain very weak system assumptions: robots are oblivious and anonymous, and their scheduling is both non-atomic and asynchronous. We would like to point out some open questions raised by our work.
First, the recent work of \cite{Devismes09} showed that for the \emph{exploration with stop} problem, randomized algorithm enabled that periodic and symmetric initial configurations are used as initial ones. However the proposed approach is not suitable for the non-atomic CORDA model. It would be interesting to consider randomized protocols for the gathering problem to bypass impossibility results.
Second, investigating the feasibility of gathering without any multiplicity detection mechanism looks challenging. Only the final configuration with a single node hosting robots could be differentiated from other configurations, even if robots are given as input the exact number of robots.

\bibliographystyle{plain}
\bibliography{biblio} 

\end{document}